\documentclass{article}

    \usepackage{iftex}
    \ifPDFTeX
    	\usepackage[T1]{fontenc}
    	\usepackage{mathpazo}
    \else
    	\usepackage{fontspec}
    \fi

    \usepackage{graphicx}
    
    \usepackage{caption}
    \DeclareCaptionFormat{nocaption}{}
    \usepackage{float}
    \usepackage{xcolor} 
    \usepackage{enumerate} 
    \usepackage{amsmath} 
    \usepackage{amssymb} 
    \usepackage{textcomp} 
    \AtBeginDocument{%
    }
    \usepackage{upquote} 
    \usepackage{eurosym} 
    \usepackage[mathletters]{ucs} 
    \usepackage{fancyvrb} 
    \usepackage{grffile} 
    \makeatletter 
    \@ifpackagelater{grffile}{2019/11/01}
    {
    }
    {
      \def\Gread@@xetex#1{%
        \IfFileExists{"\Gin@base".bb}%
        {\Gread@eps{\Gin@base.bb}}%
        {\Gread@@xetex@aux#1}%
      }
    }
    \makeatother
    \usepackage[Export]{adjustbox} 
    \adjustboxset{max size={0.9\linewidth}{0.9\paperheight}}

    \usepackage{hyperref}
    \usepackage{titling}
    \usepackage{longtable} 
    \usepackage{booktabs}  
    \usepackage[inline]{enumitem} 
    \usepackage[normalem]{ulem} 
    \usepackage{mathrsfs}
    
      \usepackage{arxiv}

    \definecolor{urlcolor}{rgb}{0,.145,.698}
    \definecolor{linkcolor}{rgb}{.71,0.21,0.01}
    \definecolor{citecolor}{rgb}{.12,.54,.11}

    \definecolor{ansi-black}{HTML}{3E424D}
    \definecolor{ansi-black-intense}{HTML}{282C36}
    \definecolor{ansi-red}{HTML}{E75C58}
    \definecolor{ansi-red-intense}{HTML}{B22B31}
    \definecolor{ansi-green}{HTML}{00A250}
    \definecolor{ansi-green-intense}{HTML}{007427}
    \definecolor{ansi-yellow}{HTML}{DDB62B}
    \definecolor{ansi-yellow-intense}{HTML}{B27D12}
    \definecolor{ansi-blue}{HTML}{208FFB}
    \definecolor{ansi-blue-intense}{HTML}{0065CA}
    \definecolor{ansi-magenta}{HTML}{D160C4}
    \definecolor{ansi-magenta-intense}{HTML}{A03196}
    \definecolor{ansi-cyan}{HTML}{60C6C8}
    \definecolor{ansi-cyan-intense}{HTML}{258F8F}
    \definecolor{ansi-white}{HTML}{C5C1B4}
    \definecolor{ansi-white-intense}{HTML}{A1A6B2}
    \definecolor{ansi-default-inverse-fg}{HTML}{FFFFFF}
    \definecolor{ansi-default-inverse-bg}{HTML}{000000}

    \definecolor{outerrorbackground}{HTML}{FFDFDF}

    \providecommand{\tightlist}{%
      \setlength{\itemsep}{0pt}\setlength{\parskip}{0pt}}
    \DefineVerbatimEnvironment{Highlighting}{Verbatim}{commandchars=\\\{\}}

    \let\Oldtex\TeX
    \let\Oldlatex\LaTeX
    \renewcommand{\TeX}{\textrm{\Oldtex}}
    \renewcommand{\LaTeX}{\textrm{\Oldlatex}}

\makeatletter
\def\PY@reset{\let\PY@it=\relax \let\PY@bf=\relax%
    \let\PY@ul=\relax \let\PY@tc=\relax%
    \let\PY@bc=\relax \let\PY@ff=\relax}
\def\PY@tok#1{\csname PY@tok@#1\endcsname}
\def\PY@toks#1+{\ifx\relax#1\empty\else%
    \PY@tok{#1}\expandafter\PY@toks\fi}
\def\PY@do#1{\PY@bc{\PY@tc{\PY@ul{%
    \PY@it{\PY@bf{\PY@ff{#1}}}}}}}
\def\PY#1#2{\PY@reset\PY@toks#1+\relax+\PY@do{#2}}

\expandafter\def\csname PY@tok@w\endcsname{\def\PY@tc##1{\textcolor[rgb]{0.73,0.73,0.73}{##1}}}
\expandafter\def\csname PY@tok@c\endcsname{\let\PY@it=\textit\def\PY@tc##1{\textcolor[rgb]{0.25,0.50,0.50}{##1}}}
\expandafter\def\csname PY@tok@cp\endcsname{\def\PY@tc##1{\textcolor[rgb]{0.74,0.48,0.00}{##1}}}
\expandafter\def\csname PY@tok@k\endcsname{\let\PY@bf=\textbf\def\PY@tc##1{\textcolor[rgb]{0.00,0.50,0.00}{##1}}}
\expandafter\def\csname PY@tok@kp\endcsname{\def\PY@tc##1{\textcolor[rgb]{0.00,0.50,0.00}{##1}}}
\expandafter\def\csname PY@tok@kt\endcsname{\def\PY@tc##1{\textcolor[rgb]{0.69,0.00,0.25}{##1}}}
\expandafter\def\csname PY@tok@o\endcsname{\def\PY@tc##1{\textcolor[rgb]{0.40,0.40,0.40}{##1}}}
\expandafter\def\csname PY@tok@ow\endcsname{\let\PY@bf=\textbf\def\PY@tc##1{\textcolor[rgb]{0.67,0.13,1.00}{##1}}}
\expandafter\def\csname PY@tok@nb\endcsname{\def\PY@tc##1{\textcolor[rgb]{0.00,0.50,0.00}{##1}}}
\expandafter\def\csname PY@tok@nf\endcsname{\def\PY@tc##1{\textcolor[rgb]{0.00,0.00,1.00}{##1}}}
\expandafter\def\csname PY@tok@nc\endcsname{\let\PY@bf=\textbf\def\PY@tc##1{\textcolor[rgb]{0.00,0.00,1.00}{##1}}}
\expandafter\def\csname PY@tok@nn\endcsname{\let\PY@bf=\textbf\def\PY@tc##1{\textcolor[rgb]{0.00,0.00,1.00}{##1}}}
\expandafter\def\csname PY@tok@ne\endcsname{\let\PY@bf=\textbf\def\PY@tc##1{\textcolor[rgb]{0.82,0.25,0.23}{##1}}}
\expandafter\def\csname PY@tok@nv\endcsname{\def\PY@tc##1{\textcolor[rgb]{0.10,0.09,0.49}{##1}}}
\expandafter\def\csname PY@tok@no\endcsname{\def\PY@tc##1{\textcolor[rgb]{0.53,0.00,0.00}{##1}}}
\expandafter\def\csname PY@tok@nl\endcsname{\def\PY@tc##1{\textcolor[rgb]{0.63,0.63,0.00}{##1}}}
\expandafter\def\csname PY@tok@ni\endcsname{\let\PY@bf=\textbf\def\PY@tc##1{\textcolor[rgb]{0.60,0.60,0.60}{##1}}}
\expandafter\def\csname PY@tok@na\endcsname{\def\PY@tc##1{\textcolor[rgb]{0.49,0.56,0.16}{##1}}}
\expandafter\def\csname PY@tok@nt\endcsname{\let\PY@bf=\textbf\def\PY@tc##1{\textcolor[rgb]{0.00,0.50,0.00}{##1}}}
\expandafter\def\csname PY@tok@nd\endcsname{\def\PY@tc##1{\textcolor[rgb]{0.67,0.13,1.00}{##1}}}
\expandafter\def\csname PY@tok@s\endcsname{\def\PY@tc##1{\textcolor[rgb]{0.73,0.13,0.13}{##1}}}
\expandafter\def\csname PY@tok@sd\endcsname{\let\PY@it=\textit\def\PY@tc##1{\textcolor[rgb]{0.73,0.13,0.13}{##1}}}
\expandafter\def\csname PY@tok@si\endcsname{\let\PY@bf=\textbf\def\PY@tc##1{\textcolor[rgb]{0.73,0.40,0.53}{##1}}}
\expandafter\def\csname PY@tok@se\endcsname{\let\PY@bf=\textbf\def\PY@tc##1{\textcolor[rgb]{0.73,0.40,0.13}{##1}}}
\expandafter\def\csname PY@tok@sr\endcsname{\def\PY@tc##1{\textcolor[rgb]{0.73,0.40,0.53}{##1}}}
\expandafter\def\csname PY@tok@ss\endcsname{\def\PY@tc##1{\textcolor[rgb]{0.10,0.09,0.49}{##1}}}
\expandafter\def\csname PY@tok@sx\endcsname{\def\PY@tc##1{\textcolor[rgb]{0.00,0.50,0.00}{##1}}}
\expandafter\def\csname PY@tok@m\endcsname{\def\PY@tc##1{\textcolor[rgb]{0.40,0.40,0.40}{##1}}}
\expandafter\def\csname PY@tok@gh\endcsname{\let\PY@bf=\textbf\def\PY@tc##1{\textcolor[rgb]{0.00,0.00,0.50}{##1}}}
\expandafter\def\csname PY@tok@gu\endcsname{\let\PY@bf=\textbf\def\PY@tc##1{\textcolor[rgb]{0.50,0.00,0.50}{##1}}}
\expandafter\def\csname PY@tok@gd\endcsname{\def\PY@tc##1{\textcolor[rgb]{0.63,0.00,0.00}{##1}}}
\expandafter\def\csname PY@tok@gi\endcsname{\def\PY@tc##1{\textcolor[rgb]{0.00,0.63,0.00}{##1}}}
\expandafter\def\csname PY@tok@gr\endcsname{\def\PY@tc##1{\textcolor[rgb]{1.00,0.00,0.00}{##1}}}
\expandafter\def\csname PY@tok@ge\endcsname{\let\PY@it=\textit}
\expandafter\def\csname PY@tok@gs\endcsname{\let\PY@bf=\textbf}
\expandafter\def\csname PY@tok@gp\endcsname{\let\PY@bf=\textbf\def\PY@tc##1{\textcolor[rgb]{0.00,0.00,0.50}{##1}}}
\expandafter\def\csname PY@tok@go\endcsname{\def\PY@tc##1{\textcolor[rgb]{0.53,0.53,0.53}{##1}}}
\expandafter\def\csname PY@tok@gt\endcsname{\def\PY@tc##1{\textcolor[rgb]{0.00,0.27,0.87}{##1}}}
\expandafter\def\csname PY@tok@err\endcsname{\def\PY@bc##1{\setlength{\fboxsep}{0pt}\fcolorbox[rgb]{1.00,0.00,0.00}{1,1,1}{\strut ##1}}}
\expandafter\def\csname PY@tok@kc\endcsname{\let\PY@bf=\textbf\def\PY@tc##1{\textcolor[rgb]{0.00,0.50,0.00}{##1}}}
\expandafter\def\csname PY@tok@kd\endcsname{\let\PY@bf=\textbf\def\PY@tc##1{\textcolor[rgb]{0.00,0.50,0.00}{##1}}}
\expandafter\def\csname PY@tok@kn\endcsname{\let\PY@bf=\textbf\def\PY@tc##1{\textcolor[rgb]{0.00,0.50,0.00}{##1}}}
\expandafter\def\csname PY@tok@kr\endcsname{\let\PY@bf=\textbf\def\PY@tc##1{\textcolor[rgb]{0.00,0.50,0.00}{##1}}}
\expandafter\def\csname PY@tok@bp\endcsname{\def\PY@tc##1{\textcolor[rgb]{0.00,0.50,0.00}{##1}}}
\expandafter\def\csname PY@tok@fm\endcsname{\def\PY@tc##1{\textcolor[rgb]{0.00,0.00,1.00}{##1}}}
\expandafter\def\csname PY@tok@vc\endcsname{\def\PY@tc##1{\textcolor[rgb]{0.10,0.09,0.49}{##1}}}
\expandafter\def\csname PY@tok@vg\endcsname{\def\PY@tc##1{\textcolor[rgb]{0.10,0.09,0.49}{##1}}}
\expandafter\def\csname PY@tok@vi\endcsname{\def\PY@tc##1{\textcolor[rgb]{0.10,0.09,0.49}{##1}}}
\expandafter\def\csname PY@tok@vm\endcsname{\def\PY@tc##1{\textcolor[rgb]{0.10,0.09,0.49}{##1}}}
\expandafter\def\csname PY@tok@sa\endcsname{\def\PY@tc##1{\textcolor[rgb]{0.73,0.13,0.13}{##1}}}
\expandafter\def\csname PY@tok@sb\endcsname{\def\PY@tc##1{\textcolor[rgb]{0.73,0.13,0.13}{##1}}}
\expandafter\def\csname PY@tok@sc\endcsname{\def\PY@tc##1{\textcolor[rgb]{0.73,0.13,0.13}{##1}}}
\expandafter\def\csname PY@tok@dl\endcsname{\def\PY@tc##1{\textcolor[rgb]{0.73,0.13,0.13}{##1}}}
\expandafter\def\csname PY@tok@s2\endcsname{\def\PY@tc##1{\textcolor[rgb]{0.73,0.13,0.13}{##1}}}
\expandafter\def\csname PY@tok@sh\endcsname{\def\PY@tc##1{\textcolor[rgb]{0.73,0.13,0.13}{##1}}}
\expandafter\def\csname PY@tok@s1\endcsname{\def\PY@tc##1{\textcolor[rgb]{0.73,0.13,0.13}{##1}}}
\expandafter\def\csname PY@tok@mb\endcsname{\def\PY@tc##1{\textcolor[rgb]{0.40,0.40,0.40}{##1}}}
\expandafter\def\csname PY@tok@mf\endcsname{\def\PY@tc##1{\textcolor[rgb]{0.40,0.40,0.40}{##1}}}
\expandafter\def\csname PY@tok@mh\endcsname{\def\PY@tc##1{\textcolor[rgb]{0.40,0.40,0.40}{##1}}}
\expandafter\def\csname PY@tok@mi\endcsname{\def\PY@tc##1{\textcolor[rgb]{0.40,0.40,0.40}{##1}}}
\expandafter\def\csname PY@tok@il\endcsname{\def\PY@tc##1{\textcolor[rgb]{0.40,0.40,0.40}{##1}}}
\expandafter\def\csname PY@tok@mo\endcsname{\def\PY@tc##1{\textcolor[rgb]{0.40,0.40,0.40}{##1}}}
\expandafter\def\csname PY@tok@ch\endcsname{\let\PY@it=\textit\def\PY@tc##1{\textcolor[rgb]{0.25,0.50,0.50}{##1}}}
\expandafter\def\csname PY@tok@cm\endcsname{\let\PY@it=\textit\def\PY@tc##1{\textcolor[rgb]{0.25,0.50,0.50}{##1}}}
\expandafter\def\csname PY@tok@cpf\endcsname{\let\PY@it=\textit\def\PY@tc##1{\textcolor[rgb]{0.25,0.50,0.50}{##1}}}
\expandafter\def\csname PY@tok@c1\endcsname{\let\PY@it=\textit\def\PY@tc##1{\textcolor[rgb]{0.25,0.50,0.50}{##1}}}
\expandafter\def\csname PY@tok@cs\endcsname{\let\PY@it=\textit\def\PY@tc##1{\textcolor[rgb]{0.25,0.50,0.50}{##1}}}

\makeatother

    \definecolor{incolor}{rgb}{0.0, 0.0, 0.5}
    \definecolor{outcolor}{rgb}{0.545, 0.0, 0.0}

    \hypersetup{
      breaklinks=true,  
      colorlinks=true,
      urlcolor=urlcolor,
      linkcolor=linkcolor,
      citecolor=citecolor,
      }
    
\begin{document}
    
    \maketitle

	\begin{abstract}

The article describes a global and arbitrage-free parametrization of the eSSVI surfaces introduced by Hendriks and Martini in 2019. A robust calibration of such surfaces has already been proposed by the quantitative research team at Zeliade in 2019, but it is sequential in expiries and lacks of a global view on the surface. The alternative calibration suggested in this article is faster and always guarantees an arbitrage-free fit of market data.

\end{abstract}

	\hypertarget{introduction}{%
\section{Introduction}\label{introduction}}

	The \emph{Stochastic Volatility Inspired} model for the total implied
variance proposed by Gatheral at the Global Derivatives conference in
Madrid in 2004 \cite{gatheral2004parsimonious} fits notably well market
data and, even restricting the \(5\) parameters domain to a Butterfly
arbitrage-free domain (as characterized by Martini and Mingone in
\cite{martini2021no}), the model still guarantees accurate smile
fitting. However, conditions for the absence of Calendar spread
arbitrage arising when smiles of different maturities are glued into a
continuous surface are still not known and this limits the practical
applications of this model.

Different sub-families of the SVI model have been studied but one in
particular has found interest in literature and industry, since it is
the extension of the SVI model to surfaces: the \emph{Surface SVI} model
introduced by Gatheral and Jacquier in \cite{gatheral2014arbitrage},
with implied total variance of the form \begin{equation}\label{eqSSVI}
w(k,T) = \frac{\theta(T)}{2}\bigl(1+\rho\varphi(T) k + \sqrt{(\varphi(T) k + \rho)^2 + (1-\rho^2)}\bigr).
\end{equation} The model has 3 parameters \((\theta, \rho, \varphi)\)
for each slice \(k\to\omega(k)\):

\begin{itemize}
\tightlist
\item
  \(\theta\) is the At-The-Money (ATM) total implied variance;
\item
  \(\rho\) is the correlation parameter, proportional to the slope of
  the smile at the ATM point;
\item
  \(\varphi\) is proportional to the ATM curvature.
\end{itemize}

The authors formulated conditions for the surface to be arbitrage-free,
in particular they found sufficient no Butterfly arbitrage conditions
(for a fixed slice) and, requiring the correlation parameter to be
constant, no Calendar Spread arbitrage conditions (among different
slices).

	Even though these surfaces perform well and are of easy implementation,
their calibration performances are not always satisfactory in practice
because of the constraint of the constant correlation parameter.
Hendriks and Martini \cite{hendriks2019extended} managed to extend the
model by formulating no Calendar Spread conditions between two SSVI
slices with different \(\rho\)s and generalized the conditions to a
continuous surface, giving birth to the \emph{extended SSVI}. The
quantitative research team at Zeliade has then proposed a convenient way
to robustly calibrate eSSVI surfaces in \cite{corbetta2019robust}. This
calibration procedure can be summarized in a nutshell as follows:

\begin{enumerate}
\def\labelenumi{\arabic{enumi}.}
\tightlist
\item
  Consider the SSVI model \cref{eqSSVI} by Gatheral and Jacquier to
  model implied variance slices at fixed maturity;
\item
  Associate to each available option maturity on the market a
  \emph{slice} \((\theta,\rho,\varphi)\) of SSVI parameters (not the
  same parameters for all slices), fulfilling the sufficient no
  Butterfly arbitrage condition obtained by Gatheral and Jacquier
  (\cite{gatheral2014arbitrage} Theorem 4.2);
\item
  Calibrate the SSVI parameters of each slice sequentially with respect
  to option maturities, in a way that ensures the absence of Calendar
  Spread arbitrage as formulated by Hendriks and Martini
  (\cite{hendriks2019extended} Proposition 3.5);
\item
  Interpolate/extrapolate \emph{linearly} the parameters; it can be
  proven that this eventually produces an arbitrage-free surface.
\end{enumerate}

All in all, such a calibrated eSSVI is parametrized by a set of
\(3 \times N\) parameters where \(N\) is the set of maturities to which
the model is calibrated, while SSVI would have \(2\times N+1\)
parameters.

The eSSVI calibration takes as underlying hypothesis the market
availability of strikes and volatilities neat the ATM point for each
slice, since total variance is constrained to go through these points.
The methodology works efficiently when market is very liquid, but when
it is not the case, approximations of the ATM volatility could cause
discrepancies between the calibrated model and real data. Furthermore,
the arbitrage-free bounds are re-written based on a first order
approximation, and they may not guarantee absence of arbitrage in case
of illiquid markets.

These observations naturally bring the necessity of a more generic and
\emph{global} procedure, which is here attained with a new calibration
algorithm, based on a re-parametrization of the no-arbitrage domain of
eSSVI as a product of intervals, suitable to be crunched by an
optimization algorithm. Parameters calibration is no more performed
sequentially slice by slice but globally on all slices. Hence, this is
why we dub \emph{Global eSSVI} the new model.

	In the present work, we firstly analyze the conditions of absence of
Calendar spread and Butterfly arbitrage for the SSVI model in
\cref{no-arbitrage-conditions-for-the-essvi-model}. Then, in
\cref{a-global-parametrization}, we define the new global
parametrization and prove it automatically satisfies the above
conditions. In \cref{calibration-strategy} we detail a calibration
algorithm and show calibration results. Numerical results are on data
from the Israel index TA35 and the exchange rate NIS/USD between the
Israeli shekel and the American dollar and they are the outcome of a
collaboration between the research team at Zeliade and Tel Aviv Stock
Exchange (TASE), to calibrate end-of-day implied volatility surfaces
with no arbitrage. We compare these results with calibration results
output by another model (that we dub Carr-Pelts-Tehranchi and present in
\cref{the-carr-pelts-tehranchi-model}). Finally, we describe a way to
interpolate and extrapolate SSVI parameters between different
maturities, in \cref{interpolation-and-extrapolation}.

	\hypertarget{no-arbitrage-conditions-for-the-essvi-model}{%
\section{No arbitrage conditions for the eSSVI
model}\label{no-arbitrage-conditions-for-the-essvi-model}}

	The SSVI model introduced by Gatheral and Jacquier
\cite{gatheral2014arbitrage} is a model for the implied total variance
\(\sigma^2_{\text{imp}}(K,T)T\) where \(\sigma_{\text{imp}}(T,K)\) is
the implied volatility for a vanilla option with strike \(K\) and
maturity \(T\), see \cref{eqSSVI}. It has been extended to the eSSVI
model by Hendriks and Martini \cite{hendriks2019extended}:
\[\text{eSSVI}(K,T) = \frac{\theta(T)}{2}\bigl(1+\rho(T)\varphi(T) k + \sqrt{(\varphi(T) k + \rho(T))^2 + (1-\rho(T)^2)}\bigr),\]
where \(k\) is the log-forward moneyness \(k=\log\frac{K}{F_0(T)}\) and
\(F_0(T)\) is the forward. From this formula, the option prices can be
recovered with the classic Black-Scholes formula \begin{align*}
C(K,T) &= D_0(T)(F_0(T)\Phi(d_1) - K\Phi(d_2)),\\
d_1 &= \frac{\sigma_{\text{imp}}(K,T)\sqrt{T}}2 - \frac{k}{\sigma_{\text{imp}}(K,T)\sqrt{T}},\\
d_2 &= d_1-\sigma_{\text{imp}}(K,T)\sqrt{T},
\end{align*} setting
\(\sigma_{\text{imp}}(K,T) = \sqrt{\frac{\text{eSSVI}(K,T)}T}\), where
\(D_0(T)\) is the discount factor and \(\Phi\) the cumulative standard
normal distribution function.

Setting \(\psi:=\theta\varphi\), the eSSVI formula in terms of the
parameters \((\theta,\rho,\psi)\) becomes
\[\text{eSSVI}(K,T) = \frac{1}{2}\bigl(\theta(T)+\rho(T)\psi(T) k + \sqrt{(\psi(T) k + \theta(T)\rho(T))^2 + \theta(T)^2(1-\rho(T)^2)}\bigr).\]

	\hypertarget{calendar-spread-arbitrage}{%
\subsection{Calendar spread arbitrage}\label{calendar-spread-arbitrage}}

	A Calendar spread arbitrage occurs when there is certainty of not losing
money and there is a non-null possibility of obtaining a positive payoff
holding two Calls with different maturities and same moneyness. In
general, the absence of such arbitrage is guaranteed by the requirement
for the Call price function to be non-decreasing in time-to-maturity for
fixed moneyness.

	Given two maturities with SSVI parameters \((\rho_1,\theta_1,\psi_1)\)
for the first and \((\rho_2,\theta_2,\psi_2)\) for the second, the no
Calendar spread arbitrage conditions have been characterized by Hendriks
and Martini in Proposition 3.5 of \cite{hendriks2019extended}:

\begin{itemize}
\tightlist
\item
  necessary conditions: \(\theta_2>\theta_1\);
  \(\psi_2>\psi_1\max\bigl(\frac{1+\rho_1}{1+\rho_2},\frac{1-\rho_1}{1-\rho_2}\bigr)\geq\psi_1\);
\item
  sufficient conditions: the necessary conditions above and
  \(\psi_2\leq\frac{\psi_1}{\theta_1}\theta_2\) (or
  \(\bigl(\rho_1-\frac{\psi_2}{\psi_1}\rho_2\bigr)^2\leq\bigl(\frac{\theta_2}{\theta_1}-1\bigr)\bigl(\frac{\psi_2^2\theta_1}{\psi_1^2\theta_2}-1\bigr)\)).
\end{itemize}

We will consider the condition
\(\psi_2\leq\frac{\psi_1}{\theta_1}\theta_2\) rather than
\(\bigl(\rho_1-\frac{\psi_2}{\psi_1}\rho_2\bigr)^2\leq\bigl(\frac{\theta_2}{\theta_1}-1\bigr)\bigl(\frac{\psi_2^2\theta_1}{\psi_1^2\theta_2}-1\bigr)\),
since it is more tractable and a natural candidate for a global
parametrization. We leave anyways open the possibility of studying a new
parametrization using the second condition.

	\hypertarget{butterfly-arbitrage}{%
\subsection{Butterfly arbitrage}\label{butterfly-arbitrage}}

	A Butterfly arbitrage arises when it is possible to make an arbitrage
from a portfolio composed of three Call options having same maturity but
different strikes. It is well-known that the absence of Butterfly
arbitrage is guaranteed if and only if the Call price function coming
from the model is convex and bounded between the discounted Call payoff
evaluated at the forward value \(D_0(T)(F_0(T)-K)^+\) and the discounted
forward \(D_0(T)F_0(T)\).

	In the case of Call prices obtained injecting an implied volatility in
the Black-Scholes formula, to avoid Butterfly arbitrage it is sufficient
to satisfy the requirement of convexity.

Given a maturity with SSVI parameters \((\rho,\theta,\psi)\), the
necessary and sufficient no Butterfly arbitrage conditions have been
described in \cite{martini2021explicit} and will be explained in
\cref{the-martini-mingone-mm-no-butterfly-arbitrage-necessary-and-sufficient-conditions}.
For efficiency reasons, it is also possible to consider a set of
sufficient but not necessary no arbitrage conditions, which are easier
to compute and implement and will be presented in
\cref{the-gatheral-jacquier-gj-no-butterfly-arbitrage-sufficient-conditions}.

	\hypertarget{the-gatheral-jacquier-gj-no-butterfly-arbitrage-sufficient-conditions}{%
\subsubsection{The Gatheral-Jacquier (GJ) no Butterfly arbitrage
sufficient
conditions}\label{the-gatheral-jacquier-gj-no-butterfly-arbitrage-sufficient-conditions}}

	These conditions are so named since they consist of the sufficient (but
not necessary) no Butterfly arbitrage conditions by Gatheral and
Jacquier \cite{gatheral2014arbitrage} in Theorem 4.2, together with the
necessary (but not sufficient) asymptotic conditions related to the
Roger Lee Moment formula \cite{lee2004moment}:

\begin{itemize}
\tightlist
\item
  necessary conditions: \(\psi\leq\frac{4}{1+|\rho|}\);
\item
  sufficient conditions: the necessary conditions above and
  \(\psi^2\leq \frac{4\theta}{1+|\rho|}:=\mathfrak f_{GJ}(\theta,|\rho|)\).
\end{itemize}

	\hypertarget{the-martini-mingone-mm-no-butterfly-arbitrage-necessary-and-sufficient-conditions}{%
\subsubsection{The Martini-Mingone (MM) no Butterfly arbitrage necessary
and sufficient
conditions}\label{the-martini-mingone-mm-no-butterfly-arbitrage-necessary-and-sufficient-conditions}}

	The name of these conditions comes from the explicit no Butterfly
arbitrage conditions in Proposition 6.3 of Martini and Mingone
\cite{martini2021explicit}. In the article, the notation used is the SVI
one, with implied total variance
\[\text{SVI}(k) = a+b(\rho(k-m)+\sqrt{(k-m)^2+\sigma^2}).\]

The SVI parameters are mapped to the SSVI parameters through
\begin{align*}
&a = \frac{\theta(1-\rho^2)}{2}, & & b=\frac{\theta\varphi}2 = \frac{\psi}2,\\
&m = -\frac{\rho}{\varphi} = -\frac{\theta\rho}{\psi}, & & \sigma=\frac{\sqrt{1-\rho^2}}{\varphi} = \frac{\theta\sqrt{1-\rho^2}}{\psi},
\end{align*} and viceversa \begin{align*}
\varphi = \frac{\sqrt{1-\rho^2}}{\sigma}, \quad \theta=\frac{2b\sigma}{\sqrt{1-\rho^2}}
\end{align*} so that \(\psi = 2b\).

The authors show that an SSVI with \(b(1+|\rho|)\leq2\), corresponding
to \(\psi\leq\frac{4}{1+|\rho|}\), automatically satisfies the necessary
Fukasawa conditions of monotonicity of the functions
\[k\to f_{1,2}(k): = \frac{k}{\sqrt{\text{SSVI}(k)}}\mp\frac{\sqrt{\text{SSVI}(k)}}2,\]
see \cite{fukasawa2012normalizing}. Then, the only additional required
condition is \begin{equation}\label{eq_sigma_MM}
\sigma\geq -\frac{bg_2(l,|\rho|)}{2(h^2(l,|\rho|)-b^2g^2(l,|\rho|))}
\end{equation} for every
\(l>l_2(|\rho|) = \tan\bigl(\frac{\arccos(-|\rho|)}{3}\bigr)^{-1}\),
where \begin{align*}
g(l,\rho) &= \frac{N'(l,\rho)}4,\\
h(l,\rho) &= 1-\Bigl(l-\frac{\rho}{\sqrt{1-\rho^2}}\Bigr)\frac{N'(l,\rho)}{2N(l,\rho)},\\
g_2(l,\rho) &= N''(l,\rho) - \frac{N'(l,\rho)^2}{2N(l,\rho)},\\
N(l,\rho) &= \sqrt{1-\rho^2}+\rho l +\sqrt{l^2+1},
\end{align*} and derivatives are taken with respect to \(l\). The
denominator in \cref{eq_sigma_MM} for \(\sigma\) is positive thanks to
the Fukasawa conditions, so that using the SSVI parameters, the
inequality becomes
\((\theta\sqrt{1-\rho^2}g^2(l,|\rho|)-g_2(l,|\rho|))\psi^2\leq 4\theta\sqrt{1-\rho^2} h^2(l,|\rho|)\),
and since \(g_2\) is negative in the considered domain for \(l\), the
inequality can be written as
\[\psi^2\leq \frac{4\theta\sqrt{1-\rho^2}h^2(l;|\rho|)}{\theta\sqrt{1-\rho^2}g^2(l;|\rho|)-g_2(l;|\rho|)}\]
for all \(l>l_2(|\rho|)\). Eventually, the necessary and sufficient no
Butterfly arbitrage conditions for SSVI read \begin{align*}
&\psi\leq\frac{4}{1+|\rho|},\\
&\psi^2 \leq \inf_{l>l_2(|\rho|)}\frac{4\theta\sqrt{1-\rho^2}h^2(l;|\rho|)}{\theta\sqrt{1-\rho^2}g^2(l;|\rho|)-g_2(l;|\rho|)} := \mathfrak f_{MM}(\theta,|\rho|).\\
\end{align*}

	\hypertarget{final-conditions}{%
\subsection{Final conditions}\label{final-conditions}}

	All in all, the Calendar spread and Butterfly constraints for successive
SSVI slices \((\rho_1,\theta_1,\psi_1)\) and
\((\rho_2,\theta_2,\psi_2)\) can be summed up to:
\begin{align}\label{eqArbitrages}
\begin{split}
& \theta_2>\theta_1>0,\\
& \psi_1 \leq \min\biggl(\frac{4}{1+|\rho_1|}, \sqrt{\mathfrak f(\theta_1,|\rho_1|)}\biggr),\\
& 0<\psi_1\max\biggl(\frac{1+\rho_1}{1+\rho_2},\frac{1-\rho_1}{1-\rho_2}\biggr)< \psi_2 \leq \min\biggl(\frac{\psi_1}{\theta_1}\theta_2, \frac{4}{1+|\rho_2|}, \sqrt{\mathfrak f(\theta_2,|\rho_2|)}\biggr)
\end{split}
\end{align} where the function \(\mathfrak f\) can be either from the MM
model (\(\mathfrak f=\mathfrak f_{MM}\)) or from the GJ model
(\(\mathfrak f=\mathfrak f_{GJ}\)).

Since the MM conditions are less strict then the GJ conditions, it could
seem natural to implement the former in a calibration routine. However,
in contrast with the latter ones, they are not explicit and require to
use a minimization algorithm to evaluate
\(\mathfrak f_{MM}(\theta,|\rho|)\), causing an increase in calibration
time.

	\hypertarget{a-global-parametrization}{%
\section{A global parametrization}\label{a-global-parametrization}}

	\hypertarget{the-case-with-two-maturities}{%
\subsection{The case with two
maturities}\label{the-case-with-two-maturities}}

	Let us consider the model with only two SSVI slices and in particular
the conditions on the second maturity parameters. In order to achieve
the condition on \(\theta_2\), we could write
\(\theta_2 = \theta_1 + \tilde a_2\) and choose \(\tilde a_2>0\). The
condition on \(\psi_2\) requires that \(\psi_2\) lies in an interval
\(]A_{\psi_2},C_{\psi_2}[\) where \begin{align*}
&A_{\psi_2} := \psi_1 p_2,\\
&C_{\psi_2} := \min\biggl(\frac{\psi_1}{\theta_1}\theta_2, f_2\biggr),\\
&p_2 := \max\biggl(\frac{1+\rho_1}{1+\rho_2},\frac{1-\rho_1}{1-\rho_2}\biggr),\\
&f_2 := \min\biggl(\frac{4}{1+|\rho_2|}, \sqrt{\mathfrak f(\theta_2,|\rho_2|)}\biggr),
\end{align*} so that setting
\(\psi_2 = c_2(C_{\psi_2}-A_{\psi_2})+A_{\psi_2}\) with \(c_2\in]0,1[\)
would guarantee the absence of arbitrage, if indeed
\(C_{\psi_2}> A_{\psi_2}\). However, the requirement
\(C_{\psi_2}> A_{\psi_2}\) is not automatically guaranteed and it
depends on the conditions of the first maturity parameters. Indeed, one
needs \(p_2<\frac{\theta_2}{\theta_1}\) and \(\psi_1 p_2<f_2\). The
first requirement is easily attained setting
\(\theta_2 = p_2\theta_1 + a_2\) with \(a_2>0\), while the second
requirement is guaranteed if and only if in the calibration of
\(\psi_1\) we also impose \(\psi_1<\frac{f_2}{p_2}\). More specifically,
\(\psi_1\) can be calibrated as \(\psi_1 = c_1 C_{\psi_1}\) where
\(c_1\in]0,1[\) and
\[C_{\psi_1} := \min\biggl(\frac{4}{1+|\rho_1|}, \sqrt{\mathfrak f(\theta_1,|\rho_1|)}, \frac{f_2}{p_2}\biggr).\]

	\hypertarget{the-general-case}{%
\subsection{The general case}\label{the-general-case}}

	The Global eSSVI is a model which uses a new global parametrization for
a set of consecutive SSVI slices, satisfying the above no arbitrage
conditions. Given \(N\) maturities, the parametrization involves
\(3\times N\) parameters as the eSSVI classical one (since it is a more
practical re-parametrization of it).

	The new parameters are \begin{equation}\label{eqProdIntervals}
\rho_1,\dots,\rho_N,\theta_1,a_2,\dots,a_N,c_1,\dots,c_N\in]-1,1[^N\times]0,\infty[^N\times]0,1[^N
\end{equation} where the \(\rho_i\) are the original eSSVI parameters
while the \(a_i\) and \(c_i\) are defined as \begin{align*}
&a_i = \theta_i-\theta_{i-1}p_i,\\
&c_i = \frac{\psi_i-A_{\psi_i}}{C_{\psi_i}-A_{\psi_i}},
\end{align*} and \begin{equation}\label{eqAuxParam}
\begin{aligned}
& p_i := \max\biggl(\frac{1+\rho_{i-1}}{1+\rho_i},\frac{1-\rho_{i-1}}{1-\rho_i}\biggr) &\ \text{if $i>1$},\\
& f_i := \min\biggl(\frac{4}{1+|\rho_i|}, \sqrt{\mathfrak f(\theta_i,|\rho_i|)}\biggr) &\ \text{if $i\geq 1$},\\
& A_{\psi_1} := 0 &\ \text{if $i=1$},\\
& A_{\psi_i} := \psi_{i-1}p_i &\ \text{if $i>1$},\\
& C_{\psi_1}:=\min\biggl(f_1,\frac{f_2}{p_2},\dots,\frac{f_N}{\prod_{j=2}^Np_j}\biggr) &\ \text{if $i=1$},\\
& C_{\psi_i}:=\min\biggl(\frac{\psi_{i-1}}{\theta_{i-1}}\theta_i,f_i,\frac{f_{i+1}}{p_{i+1}},\dots,\frac{f_N}{\prod_{j={i+1}}^Np_j}\biggr) &\ \text{if $i>1$}.
\end{aligned}
\end{equation}

The original SSVI parameters are sequentially obtained through the
ordered relations \begin{equation}\label{eqSSVItoGlobal}
\begin{aligned}
&\theta_2 = \theta_1p_2+a_2,\quad \dots,\quad \theta_N=\theta_{N-1}p_N+a_N,\\
&\psi_1 = c_1(C_{\psi_1}-A_{\psi_1})+A_{\psi_1},\quad \dots,\quad \psi_N = c_N(C_{\psi_N}-A_{\psi_N})+A_{\psi_N}.
\end{aligned}
\end{equation}

	It can be useful to note that:

\begin{itemize}
\tightlist
\item
  \(p_i \geq 1\) with \(p_i=1\) iff \(\rho_i=\rho_{i-1}\);
\item
  \(0<f_i \leq 4\).
\end{itemize}

	The following result guarantees that the Global eSSVI parametrization
has no arbitrage.

	\begin{proposition}

For any integer $N>0$ and any parameters
$$\rho_1,\dots,\rho_N,\theta_1,a_2,\dots,a_N,c_1,\dots,c_N\in]-1,1[^N\times]0,\infty[^N\times]0,1[^N,$$
the set of $N$ SSVI slices
$$\text{eSSVI}_i(K,T_i) = \frac{1}{2}\bigl(\theta_i+\rho_i\psi_i k + \sqrt{(\psi_i k + \theta_i\rho_i)^2 + \theta_i^2(1-\rho_i^2)}\bigr),$$
with today's forward $F_0(T_i)$ for increasing maturities $T_i$, log-forward moneyness $k=\log\frac{K}{F_0(T_i)}$ and parameters $(\theta_i,\rho_i,\psi_i)$ defined through \cref{eqAuxParam} and \cref{eqSSVItoGlobal}, is free of Butterfly and Calendar spread arbitrage.

\end{proposition}

	\begin{proof}
We now prove that with such parametrization the arbitrage constraints in \cref{eqArbitrages} are satisfied. Indeed,
\begin{itemize}
\item $\theta_1>0$ is chosen from the start, and $\theta_i=\theta_{i-1}p_i+a_i>\theta_{i-1}p_i\geq\theta_{i-1}$;
\item $\psi_1=c_1C_{\psi_1}>0$ and $\psi_1\leq c_1f_1<f_1$ (by assumption on the domain of $c_1$), which is the no Butterfly arbitrage for the first maturity.
\end{itemize}

We need to show that if $\psi_i=c_i(C_{\psi_i}-A_{\psi_i})+A_{\psi_i}$, then
\begin{equation*}
A_{\psi_i}<\psi_i\leq \min\biggl(\frac{\psi_{i-1}}{\theta_{i-1}}\theta_i,f_i\biggr).
\end{equation*}
We have already seen that this holds true for $i=1$, where it holds even that $A_{\psi_1}<\psi_1< C_{\psi_1}$. We now show by induction that if $A_{\psi_{i-1}}<\psi_{i-1}<C_{\psi_{i-1}}$ then $A_{\psi_i}<C_{\psi_i}$. This will ensure the no arbitrage condition, since then, by definition, $\psi_i=c_i(C_{\psi_i}-A_{\psi_i})+A_{\psi_i}$ so that $A_{\psi_i}<\psi_i< C_{\psi_i}$ and also $C_{\psi_i}=\min\Bigl(\frac{\psi_{i-1}}{\theta_{i-1}}\theta_i,f_i,\frac{f_{i+1}}{p_{i+1}},\dots,\frac{f_N}{\prod_{j={i+1}}^Np_j}\Bigr)\leq\min\bigl(\frac{\psi_{i-1}}{\theta_{i-1}}\theta_i,f_i\bigr)$.

For induction, we suppose we have proven $A_{\psi_{i-1}}<C_{\psi_{i-1}}$. If $C_{\psi_i}=\frac{\psi_{i-1}}{\theta_{i-1}}\theta_i$, the requirement $\theta_i>\theta_{i-1}p_i$ implies $A_{\psi_i}=\psi_{i-1}p_i<C_{\psi_i}$. Otherwise, the inequality holds true iff $\psi_{i-1}<\frac{1}{p_i}\min\bigl(f_i,\frac{f_{i+1}}{p_{i+1}},\dots,\frac{f_N}{\prod_{j={i+1}}^Np_j}\bigr)$. By the above consequence of the induction hypothesis,
\begin{align*}
\psi_{i-1}&<C_{\psi_{i-1}} = \min\biggl(\frac{\psi_{i-2}}{\theta_{i-2}}\theta_{i-1},f_{i-1},\frac{f_{i}}{p_{i}},\dots,\frac{f_N}{\prod_{j={i}}^Np_j}\biggr)\\
&\leq\min\biggl(\frac{f_i}{p_i},\frac{f_{i+1}}{p_ip_{i+1}},\dots,\frac{f_N}{\prod_{j={i}}^Np_j}\biggr) \\
&=\frac{1}{p_i}\min\biggl(f_i,\frac{f_{i+1}}{p_{i+1}},\dots,\frac{f_N}{\prod_{j={i+1}}^Np_j}\biggr)
\end{align*}
and this concludes the proof.
\end{proof}

	\hypertarget{calibration-strategy}{%
\section{Calibration strategy}\label{calibration-strategy}}

	The Global eSSVI parametrization can be easily implemented for
calibration purposes. The calibration function can either target the
market total implied variance or the market option prices. Indeed, model
prices can be easily recovered through the Black-Scholes formula, given
implied volatility \(\sqrt{\frac{\text{eSSVI}(K,T)}T}\).

We performed tests on real market data on two different assets: the
Israel index TA35 and the exchange rate NIS/USD between the Israeli
shekel and the American dollar. We describe the procedure which led to
very satisfactory calibration results.

	\hypertarget{routine}{%
\subsection{Routine}\label{routine}}

	The eSSVI calibration consists of finding the parameters
\(\{\rho_i\}_{i=1}^N,\theta_1,\{a_i\}_{i=2}^N,\{c_i\}_{i=1}^N\) such
that the model eSSVI prices best match the basket of available option
prices. In our experiments, parameters are chosen to minimize the
possibly weighted squared difference between market option prices
\(\tilde C(K,T)\) and model option prices \(C(K,T)\):
\[\sum_{K, T}\bigl(\tilde C(K,T)-C(K,T)\bigr)^2\omega(K,T).\] Weights
\(\omega\) can be arbitrarily chosen by the user. In particular, they
can be chosen to match the inverse of squared market Black-Scholes
vegas, in order to have a calibration in implied volatilities at the
first order, instead of a calibration in prices.

The calibration is performed using the \verb|least_squares| function in
the \verb|scipy.optimize| library. The maximum number of function
evaluations is set at \(1000\) and the argument for convergence
\verb|ftol| is set at its default value \(10^{-8}\).

	The algorithm works as follows:

\begin{enumerate}
\def\labelenumi{\arabic{enumi}.}
\tightlist
\item
  input:
  \(\{\rho_i\}_{i=1}^N,\theta_1,\{a_i\}_{i=2}^N,\{c_i\}_{i=1}^N\);
\item
  with input parameters \(\rho_1,\dots,\rho_N\) (they are also model
  parameters) compute the intermediate quantities \(p_2,\dots,p_N\) in
  \cref{eqAuxParam};
\item
  with input parameters \(\theta_1,a_2,\dots,a_N\) compute from
  \cref{eqSSVItoGlobal} and step 2. the model parameters
  \(\theta_2,\dots,\theta_N\);
\item
  compute from the \(\theta_i\) and the \(\rho_i\) the intermediate
  quantities \(f_1,\dots,f_N\);
\item
  compute the intermediate quantity \(C_{\psi_1}\);
\item
  with input parameter \(c_1\) compute the model parameter \(\psi_1\)
  from \(c_1\), \(C_{\psi_1}\) and \(A_{\psi_1}=0\);
\item
  then, for each \(i\) increasing from \(2\) to \(N\):

  \begin{itemize}
  \tightlist
  \item
    with input parameter \(c_i\) compute the intermediate quantities
    \(A_{\psi_i}\), \(C_{\psi_i}\),
  \item
    compute the model parameter \(\psi_i\);
  \end{itemize}
\item
  for each option with strike \(K\) and maturity \(T_i\), compute the
  model price \(C(K,T_i)\);
\item
  evaluate
  \(\sum_{K, T_i}\bigl(\tilde C(K,T_i)-C(K,T_i)\bigr)^2\omega(K,T_i)\)
  and repeat all the steps with new parameters
  \(\{\rho_i\}_{i=1}^N,\theta_1,\{a_i\}_{i=2}^N,\{c_i\}_{i=1}^N\) in
  order to minimize this quantity.
\end{enumerate}

	Recall that the minimum search can be performed over the product of
intervals \cref{eqProdIntervals}, which is an appealing feature when
using standard minimization routines such as the \verb|least_squares|
function.

In general, trades in the market have different timestamps, so that
available option prices are not simultaneous. Then, in point 8., the
model price is evaluated using the forward computed at the timestamp of
the corresponding market option, so it could be possibly different for
different options. This is linked to the fact that model parameters are
supposed to be constant in logforward-moneyness. In particular, given a
market option with strike \(K\) and maturity \(T\) traded (or quoted) at
timestamp \(t\), the corresponding model price is: \begin{align*}
C(K,T) &= D_C(T)(F_t(T)\Phi(d_1) - K\Phi(d_2)),\\
d_1 &= \frac{\sigma_{\text{imp}}(K,T)\sqrt{T-t}}2 - \frac{k}{\sigma_{\text{imp}}(K,T)\sqrt{T-t}},\\
d_2 &= d_1-\sigma_{\text{imp}}(K,T)\sqrt{T-t},\\
\sigma_{\text{imp}}(K,T)\sqrt{T-t} &= \sqrt{\frac{\theta(T)+\rho(T)\psi(T) k + \sqrt{(\psi(T) k + \theta(T)\rho(T))^2 + \theta(T)^2(1-\rho(T)^2)}}{2}},\\
k &= \log\frac{K}{F_t(T)},
\end{align*} where \(D_C(T)\) is the discount factor at closing time
(here we suppose it does not change a lot during the day) and
\(F_t(T)=F_C(T)\frac{S_t}{S_C}\) is the current forward.

	\hypertarget{parameters-domain-and-initial-conditions}{%
\subsection{Parameters domain and initial
conditions}\label{parameters-domain-and-initial-conditions}}

	We set the following initial conditions:

\begin{itemize}
\tightlist
\item
  the \(a\) parameters are obtained guessing initial values for the
  \(\theta\)s from the ATM total implied variances;
\item
  initial \(\rho\)s are set to the intermediate value \(0\);
\item
  initial \(c\)s are set to \(0.5\).
\end{itemize}

	The \(a\) parameters could lie in an infinite range but, for
optimization reasons, it is a good practice to bound them. After many
tests on TA35 and NIS/USD data, we chose to impose \(a\) to be smaller
than \(0.05\). If the maximum of the initial values for the \(a\)s is
larger than the fixed upper bound, we double the latter bound. The other
parameters are already bounded from the definition of the Global eSSVI
model, indeed \(\rho\in]-1,1[\) and \(c\in]0,1[\).

The calibration weights \(\omega(K,T)\) are taken to be constant.

	\hypertarget{numerical-experiments}{%
\subsection{Numerical experiments}\label{numerical-experiments}}

	We show here the numerical results obtained with the GJ conditions for
the Global eSSVI parametrization and compare them with the well-know
rich and flexible parametric price surface of Carr and Pelts, which we
describe in the following section.

	\hypertarget{the-carr-pelts-tehranchi-model}{%
\subsubsection{The Carr-Pelts-Tehranchi
model}\label{the-carr-pelts-tehranchi-model}}

	Carr and Pelts presented in 2015, at a conference in honour of Steven
Shreve at Purdue university, an explicit arbitrage-free parametrization
for FX option prices, which at the time seemed to go completely under
the academic and practitioner radars. In 2019, in a deep and brilliant
paper on a subtle property of the Black-Scholes formula
\cite{tehranchi2020black}, Mike Tehranchi re-discovered independently
this family of models, with the more mathematical perspective of
semi-groups acting on sets of convex functions (Calls and Puts
normalized prices). Therefore, from now on we name this model using the
acronym CPT.

Unlike eSSVI, CPT gives a direct formula for the vanilla \emph{price},
not its implied volatility. The implied volatility is not a natural
object in the CPT family (except of course in the case of the
Black-Scholes model itself, which indeed belongs to this family); if one
needs to get the implied volatility, it is required to resort to
numerical algorithm like the excellent \emph{rationale} approach by
Jaeckel.

	In order to perform calibrations, we use the approach accurately
described in section 3 of \cite{antonov2019new}. Denote \(S_t\) the
current value of the underlyer (in case of dividends and rates not null,
\(S_t\) should be replaced by the Forward \(F_t\)) and pick up a
\emph{log-concave density} \(f=e^{-h}\) on \(\mathbb R\). Then, under
the risk-neutral measure, the law of the underlyer \(S\) at fixed
maturity will be given by \[S^{(\tau)}:= S_t \frac{f(\tau+Z)}{f(Z)}\]
where \(Z\) is a random variable with law \(f\), and \(\tau\) some real
positive parameter. It holds
\[E[(S^{(\tau)}-K)^+] = S_t \int_R f(\tau+z) dz -K \int_R f(z) dz\]
where
\(R:= \Bigl\{z\;\big| \frac{f(\tau+z)}{f(z)}> \frac{K}{S_t} \Bigr\}\).
Observe now that \(z \to \frac{f(\tau+z)}{f(z)}\) is non-decreasing,
indeed
\(\frac{d}{dz}\bigl(\log{\frac{f(\tau+z)}{f(z)}}\bigr)=-h'(\tau+z)+h'(z)\)
where \(h=-\log{f}\) is convex (so its second derivative is positive).
It follows that
\[d_f(\tau,k):=\sup\Bigl\{ z \; \big| \log{\frac{f(\tau+z)}{f(z)}}=k\Bigr\} = \sup\bigl\{ z \; \big| h(\tau+z)-h(z)=-k\bigr\}\]
is well defined and that \(R=\{z \leq d_f(\tau,k)\}\), where
\(k=\log\bigl(\frac{K}{S_t}\bigr)\). Eventually
\[E[(S^{(\tau)}-K)^+] = S_t\Omega(d_f(\tau,k)+\tau) -K \Omega(d_f(\tau,k))\]
where \(\Omega\) is the cumulative density function of \(f\).

All the above \emph{massively generalizes} the Black-Scholes formula,
which corresponds to the particular case
\(f(x)=\frac{1}{\sqrt{2 \pi}}\exp\bigl(\frac{-x^2}{2}\bigr)\).

The second ingredient to CPT is based on the remarks that:

\begin{enumerate}
\def\labelenumi{\arabic{enumi}.}
\tightlist
\item
  if \(\tau<z\), then for every \(K>0\),
  \(E[(S^{(\tau)}-K)^+]<E[(S^{(z)}-K)^+]\).
\item
  \(S_t \Omega(d_f(\tau,k)+\tau) -K \Omega(d_f(\tau,k))=(S_t-K)^+\) iff
  \(\tau=0\).
\end{enumerate}

It follows that if one chooses any non-decreasing continuous function
\(T \to \tau(T)\) such that \(\tau(0)=0\), then the \emph{price
surface}:
\[(K,T) \to S_t \Omega(d_f(\tau(T),k)+\tau(T)) - K\Omega(d_f(\tau(T),k))\]
is free of arbitrage.

	In Antonov et al.~specification, the function \(\tau\) is a
piecewise-linear function and \(h\) a piecewise-quadratic differentiable
convex function. It is calibrated using a grid of \(2N_{CPT}\) node
points. Then, the algorithm requires to make a choice on the number of
nodes used to calibrate the model density. There is always a trade-off
between taking \(N_{CPT}\) large, which could allow for an increased
fitting ability at the price of more instability in the results if there
are too few options in the calibration basket, and choosing a smaller
\(N_{CPT}\) with the opposite benefits or issues. A rational start is to
compare eSSVI and CPT number of parameters; eSSVI has \(3N\) parameters,
while CPT has \(N+2N_{CPT}\) parameters, so equating them gives
\(N_{CPT}=N\). In practice, on both the TA35 market and the NISUSD one,
the choice of \(N_{CPT}=6\) gives very good fit results and, since we
generally have no more than \(6\) maturities, it is even too accurate
when comparing with the number of eSSVI parameters.

	\hypertarget{price-and-volatility-plots}{%
\subsubsection{Price and volatility
plots}\label{price-and-volatility-plots}}

	The calibration algorithm has as a target the market prices. First, we
show the Call and Put prices on the date \(2021/10/26\) for both the
TA35 index (spot of \(1871.67\)) and the NIS/USD Forex (spot of
\(319.98\)). The calibration basket is composed of both trade and quote
prices in the last \(10\) minutes of trading, filtered to remove noisy
data and aggregated to a synthetic market price per option. The GJ
Global eSSVI and the CPT models are calibrated and the corresponding
model prices are shown in \Cref{figurePricesTA35,figurePricesNISUSD}
with the following notation:

\begin{itemize}
\tightlist
\item
  the marker \(\times\) indicates a quote;
\item
  the marker \(\cdot\) indicates a trade;
\item
  the vertical line indicates the bid-ask prices;
\item
  a black marker indicates a model price outside the bid-ask.
\end{itemize}


    \begin{figure}
    	\adjustimage{max size={0.9\linewidth}{0.9\paperheight}}{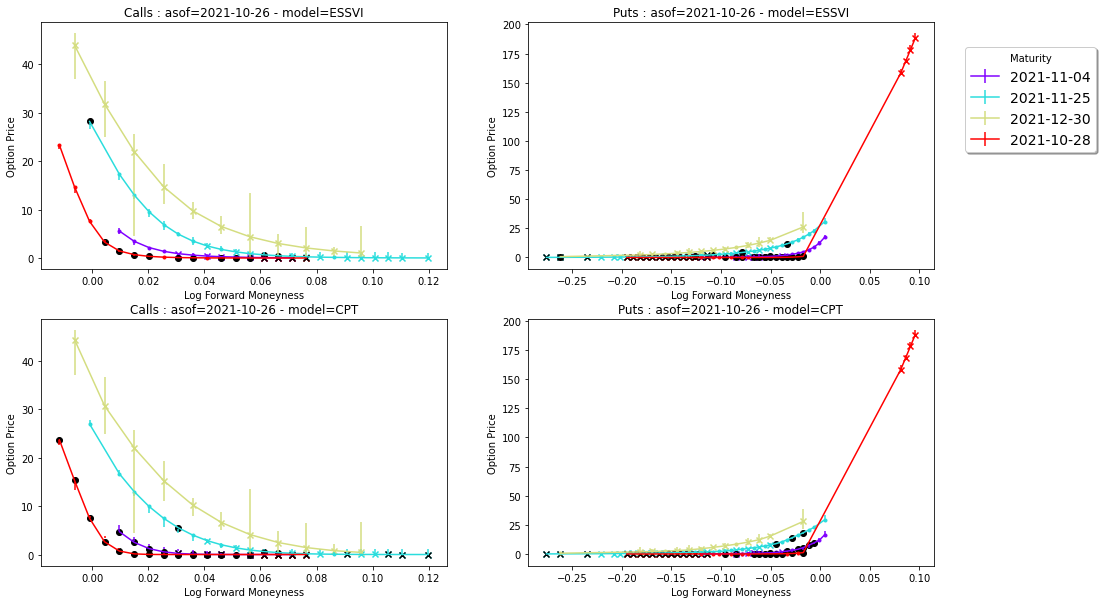}
    	\caption{Comparison of calibrated Call and Put prices on TA35 with the eSSVI and the CPT models.
    		Markers $\times$ indicate quotes, markers $\cdot$ indicate trades, black markers indicate model prices outside the bid-ask.}
    	\label{figurePricesTA35}
    \end{figure}
    { \hspace*{\fill} \\}


    \begin{figure}
    	\adjustimage{max size={0.9\linewidth}{0.9\paperheight}}{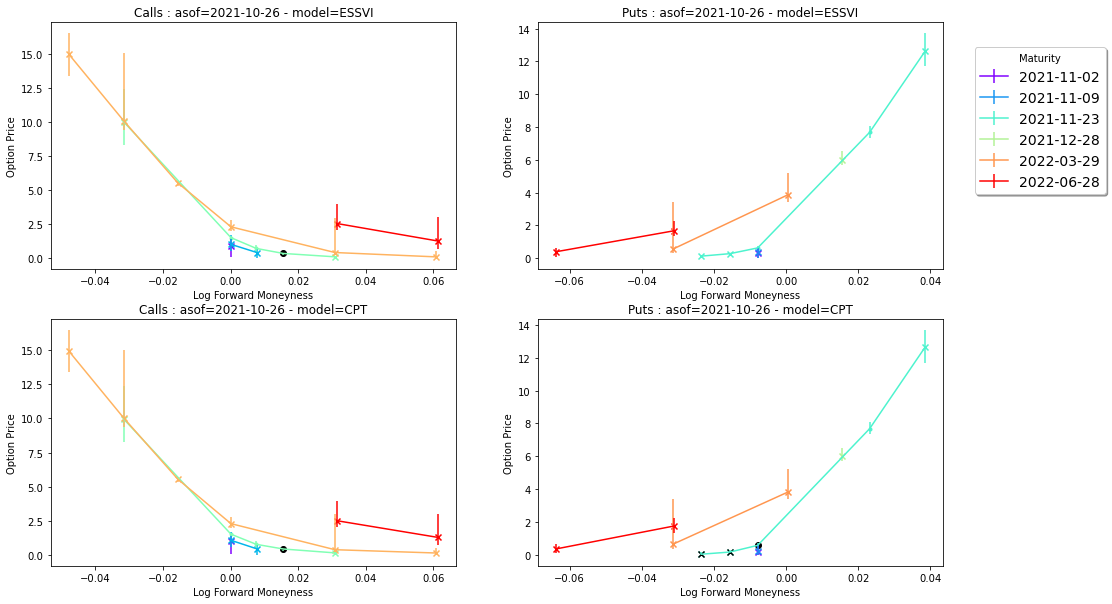}
    	\caption{Comparison of calibrated Call and Put prices on NIS/USD with the eSSVI and the CPT models. Markers $\times$ indicate quotes, markers $\cdot$ indicate trades, black markers indicate model prices outside the bid-ask.}
    	\label{figurePricesNISUSD}
    \end{figure}
    { \hspace*{\fill} \\}

	The corresponding absolute errors in basis point to the Forward are
shown in \Cref{figErrorTA35,figErrorNISUSD}.


    \begin{figure}
    	\adjustimage{max size={0.9\linewidth}{0.9\paperheight}}{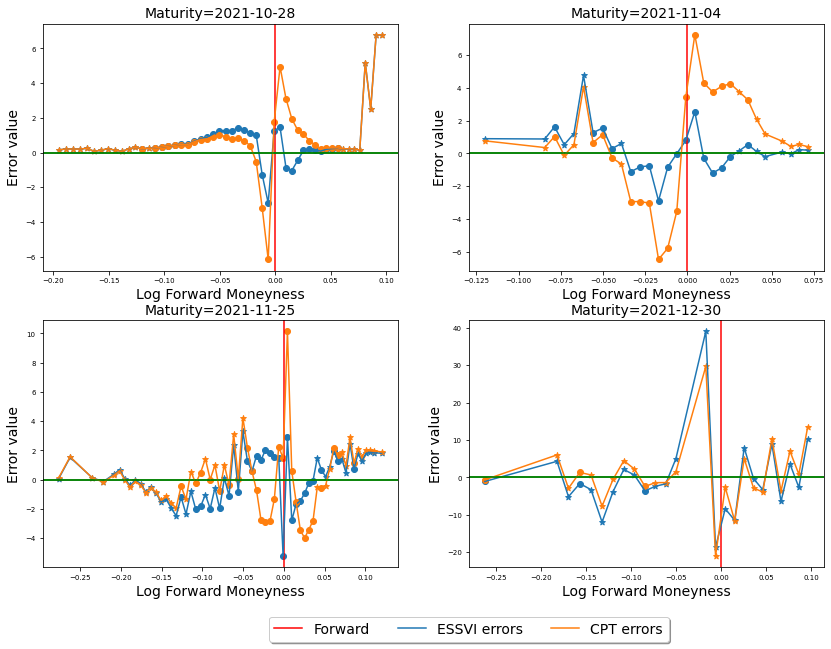}
    	\caption{Comparison of calibration errors in basis points to the Forward for OTM options, for TA35 in the eSSVI and the CPT models.}
    	\label{figErrorTA35}
    \end{figure}
    { \hspace*{\fill} \\}


    \begin{figure}
    	\adjustimage{max size={0.9\linewidth}{0.9\paperheight}}{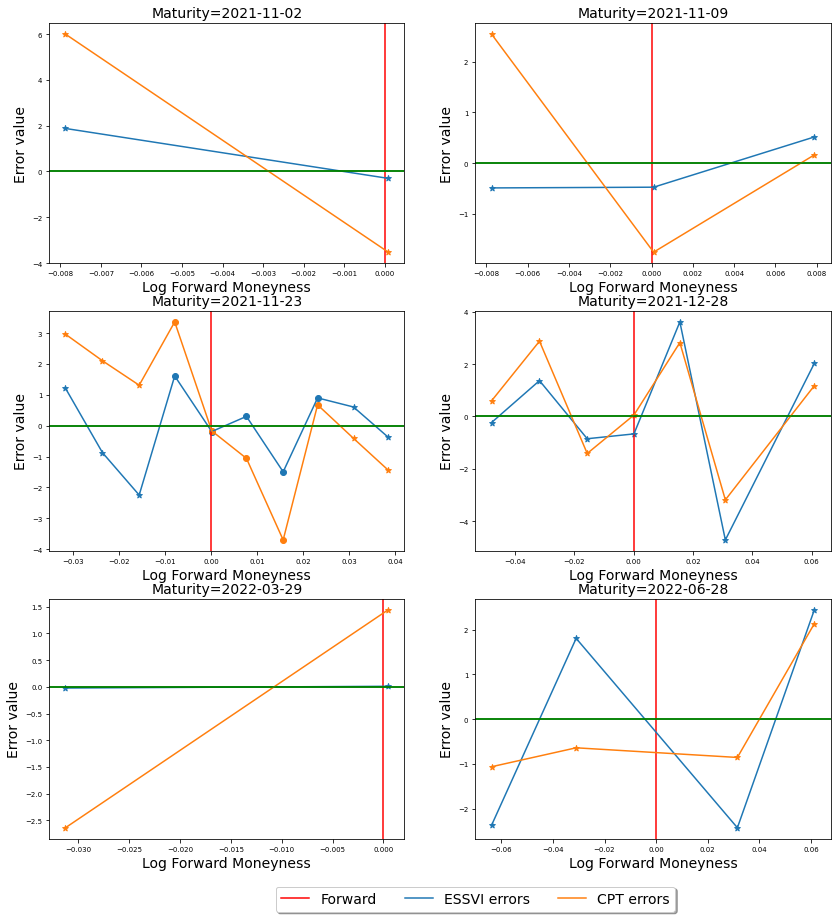}
    	\caption{Comparison of calibration errors in basis points to the Forward for OTM options, for NIS/USD in the eSSVI and the CPT models.}
    	\label{figErrorNISUSD}
    \end{figure}
    { \hspace*{\fill} \\}

	For the above plots, we also report the corresponding implied
volatilities in \Cref{figIVTA35,figIVNISUSD}.


    \begin{figure}
    	\adjustimage{max size={0.9\linewidth}{0.9\paperheight}}{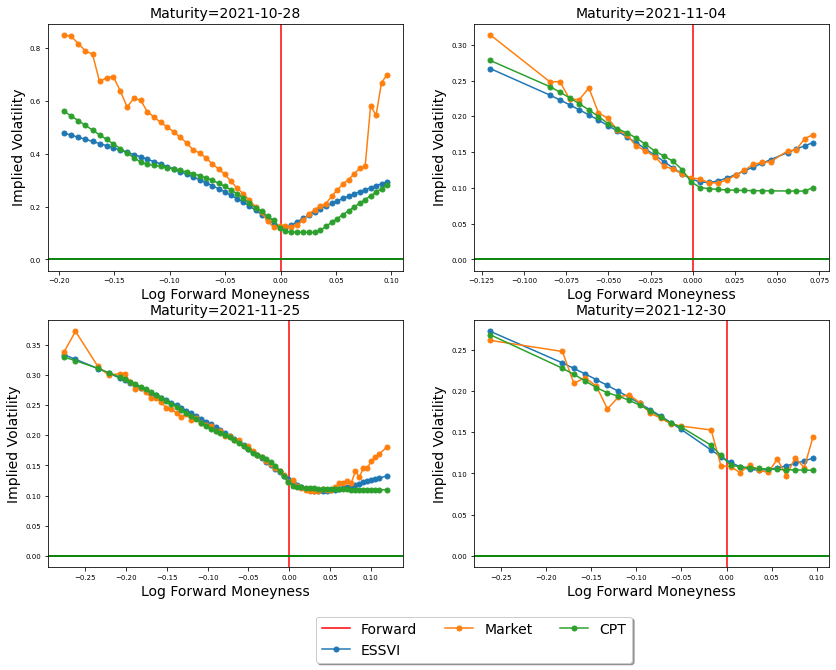}
    	\caption{Comparison of calibrated smiles for TA35 in the eSSVI and the CPT models.}
    	\label{figIVTA35}
    \end{figure}
    { \hspace*{\fill} \\}


    \begin{figure}
    	\adjustimage{max size={0.9\linewidth}{0.9\paperheight}}{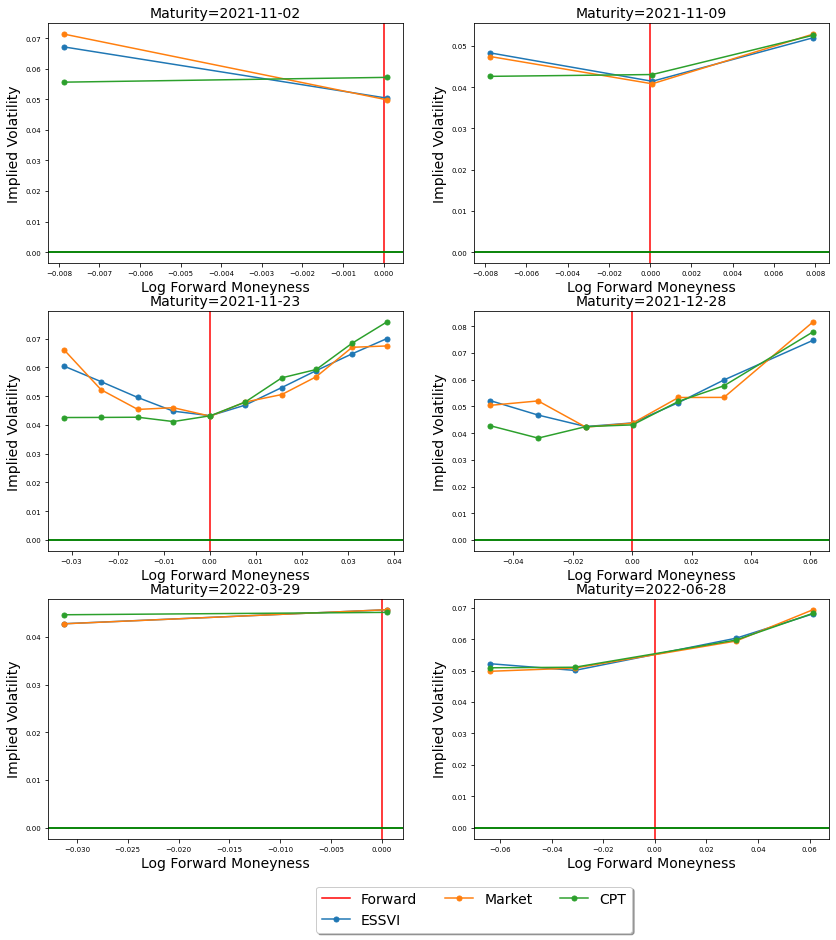}
    	\caption{Comparison of calibrated smiles for NIS/USD in the eSSVI and the CPT models.}
    	\label{figIVNISUSD}
    \end{figure}
    { \hspace*{\fill} \\}
    
	Results are very satisfactory and we can see that the two calibration
results are comparable if not even better for the GJ Global eSSVI. The
calibration of the shortest maturity is typically more difficult in
terms of error reduction, which explains the high number of black dots
in the Call and Put prices graphs. However, the ATM prices are always
well calibrated.

In general, the GJ Global eSSVI calibration is preferable to the CPT for
four reasons:

\begin{itemize}
\tightlist
\item
  from a theoretical point of view, it is easier to understand the role
  of its parameters and how to tune them in order to change the smile
  shape;
\item
  it directly models a volatility surface, so that the calibrated
  implied smiles have more natural and desirable shapes. Also, working
  with implied volatilities allows to compare between underlyers and
  dates, while directly comparing prices is less obvious;
\item
  it is much easier and straightforward to be coded;
\item
  calibration times are \(10\) times smaller than CPT calibration times.
\end{itemize}

	\hypertarget{no-arbitrage-check}{%
\subsection{No arbitrage check}\label{no-arbitrage-check}}

	As a sanity check of the code implementation, we implemented an
arbitrage identifier routine as described in Section 3 of
\cite{cohen2020detecting}. It is worth noticing that such routine is not
necessary since models have been shown to be arbitrage-free. However,
numerical approximations could arise in computed prices, even though
calibration parameters live in the no arbitrage domain.

We take all model prices resulting from the calibration routine and
verify (with linear constraints) whether there is any kind of arbitrage.
In particular, we check the positivity of prices and look for what
Reisinger et al.~\cite{cohen2020detecting} call Vertical Spread,
Vertical Butterfly, Calendar Spread, Calendar Vertical Spread, and
Calendar Butterfly arbitrage.

Results confirm the lack of arbitrage in prices resulting from both the
Global eSSVI and the CPT models. Some little arbitrage caused by
numerical approximations could arise with the former model for values of
\(\rho\) near \(\pm1\). A way to avoid this is to bound \(\rho\) in a
smaller interval, such as \(]-0.95,0.95[\). Calibration results in terms
of overall calibration error do not suffer from this choice.

	\hypertarget{interpolation-and-extrapolation}{%
\section{Interpolation and
extrapolation}\label{interpolation-and-extrapolation}}

	This section describes how to interpolate and extrapolate eSSVI
parameters on maturities which differ from the ones used in the
calibration, in order to guarantee the absence of arbitrage. These
methodologies are taken from \cite{corbetta2019robust}.

	\hypertarget{interpolation}{%
\subsection{Interpolation}\label{interpolation}}

	Suppose we have calibrated the model on two maturities \(T_1<T_2\) with
eSSVI parameters \((\rho_1,\theta_1,\psi_1)\) for the first one and
\((\rho_2,\theta_2,\psi_2)\) for the second. How could we interpolate
arbitrage-free parameters for a maturity \(t\in[T_1,T_2]\)? Similarly to
what is done in section 5.1.2 of \cite{corbetta2019robust}, we define
the new parameters \((\rho_t,\theta_t,\psi_t)\) through the scheme:

\begin{itemize}
\tightlist
\item
  \(\theta_{t} = (1-\lambda)\theta_1 + \lambda\theta_2\)
\item
  \(\psi_{t} = (1-\lambda)\psi_1 + \lambda\psi_2\)
\item
  \(\psi_{t}\rho_t = (1-\lambda)\psi_1\rho_1 + \lambda\psi_2\rho_2\)
\end{itemize}

where \(\lambda = \frac{t-T_1}{T_2-T_1}\). We now show that in such way
arbitrage conditions are satisfied.

	We look at the Calendar Spread conditions and take
\(T_1\leq t<u\leq T_2\) with \(\lambda = \frac{t-T_1}{T_2-T_1}\) and
\(\mu = \frac{u-T_1}{T_2-T_1}\). First, it is immediate that
\(\theta_u>\theta_t\). Second, the sufficient condition
\(\psi_u\theta_t-\psi_t\theta_u\leq0\) can be rewritten as
\((\mu-\lambda)(\psi_2\theta_1-\psi_1\theta_2)\leq0\), which is verified
since the two calibrated slices are arbitrage-free. Last, the second
necessary condition is equivalent to
\(\psi_u(1\pm\rho_u)-\psi_t(1\pm\rho_t)>0\). Substituting and
simplifying as above, we find
\((\mu-\lambda)(\psi_2(1\pm\rho_2)-\psi_1(1\pm\rho_1))>0\), which again
holds true.

	The Butterfly arbitrage conditions at time \(t\) are proven in
\cite{corbetta2019robust}.

	\hypertarget{extrapolation}{%
\subsection{Extrapolation}\label{extrapolation}}

	The extrapolation procedure is also taken from
\cite{corbetta2019robust}, sections 5.2 and 5.3. In the following, for
completeness we give full proofs of the consistency of this choice.

	\hypertarget{before-the-1st-maturity}{%
\subsubsection{Before the 1st maturity}\label{before-the-1st-maturity}}

	For \(t\) smaller than the first calibrated maturity \(T_1\), we set

\begin{itemize}
\tightlist
\item
  \(\theta_t = \lambda\theta_1\)
\item
  \(\psi_t = \lambda\psi_1\)
\item
  \(\rho_t = \rho_1\)
\end{itemize}

where \(\lambda=\frac{t}{T_1}<1\).

	It is easy to verify that the fact that the parameters \(\theta\) and
\(\psi\) are increasing while \(\rho\) is constant, combined with the
arbitrage-free calibration of these parameters on the first maturity,
guarantees the absence of arbitrage for the extrapolated triple
\((\theta_t,\rho_t,\psi_t)\). Indeed, the absence of arbitrage is easily
checked observing that the implied volatility on maturities \(t\) and
\(T_1\) is the same for fixed \(k\), so that total variances are
increasing with respect to maturity and cannot have Butterfly arbitrage,
given that the smile on \(T_1\) is Butterfly arbitrage-free.

To show that the implied volatilities coincide, we look at the eSSVI
formula on maturity \(T_1\)
\[\text{eSSVI}(K,T_1) = \frac{1}{2}\bigl(\theta_1+\rho_1\psi_1 k + \sqrt{(\psi_1 k + \theta_1\rho_1)^2 + \theta_1^2(1-\rho_1^2)}\bigr)\]
and on maturity \(t\)
\[\text{eSSVI}(\tilde K,t) = \frac{1}{2}\bigl(\lambda\theta_1+\rho_1\lambda\psi_1\tilde k + \sqrt{(\lambda\psi_1\tilde k + \lambda\theta_1\rho_1)^2 + (\lambda\theta_1)^2(1-\rho_1^2)}\bigr).\]
It is easy to see that if \(k=\tilde k\), corresponding to
\(\tilde K = K\frac{F_0(t)}{F_0(T_1)}\), then
\(\text{eSSVI}(\tilde K,t) = \lambda \text{eSSVI}(K,T_1) = \frac{t}{T_1}\text{eSSVI}(K,T_1)\)
and the conclusion immediately follows.

	\hypertarget{after-the-last-maturity}{%
\subsubsection{After the last maturity}\label{after-the-last-maturity}}

	Extrapolation on the right of the last calibrated maturity \(T_N\) is
performed setting

\begin{itemize}
\tightlist
\item
  \(\theta_t = \theta_N + \frac{\theta_N-\theta_{N-1}}{T_N-T_{N-1}}(t-T_N)\)
\item
  \(\psi_t = \psi_N\)
\item
  \(\rho_t = \rho_N\)
\end{itemize}

where \(t>T_N\). The first bullet point is such that \(\theta_t\)
preserves the last slope available, but it can be replaced with any
positive angular coefficient.

	Since \(\psi\) and \(\rho\) are constant, \(\theta\) is increasing and
the parameters satisfy no arbitrage conditions at maturity \(T_N\), it
is straightforward to show that the extrapolated parameters have no
arbitrage.

	\hypertarget{acknowledgments}{%
\section{Acknowledgments}\label{acknowledgments}}

	The findings in this paper and their implementation would not have been
possible without the key support of Zeliade team and especially Claude
Martini, who supported the project and laid the foundations for the
results, Pierre Cohort, who implemented the Python code, and Ismail
Laachir, who took care of the whole data manipulation.

I thank Stefano De Marco for the detailed and patient reading of the
article, and for the improvements suggested.

I also thank TASE for providing the data displayed in the present
article and in particular Or Amir and his team for technical suggestions
on the parameters setup.


\newpage \bibliography{Biblio}

\begin{thebibliography}{10}

\bibitem{antonov2019new}
Alexandre Antonov, Michael Konikov, and Michael Spector.
\newblock {A new arbitrage-free parametric volatility surface}.
\newblock {\em Available at SSRN 3403708}, 2019.

\bibitem{cohen2020detecting}
Samuel~N Cohen, Christoph Reisinger, and Sheng Wang.
\newblock {Detecting and repairing arbitrage in traded option prices}.
\newblock {\em Applied Mathematical Finance}, 27(5):345--373, 2020.

\bibitem{corbetta2019robust}
Jacopo Corbetta, Pierre Cohort, Ismail Laachir, and Claude Martini.
\newblock {Robust calibration and arbitrage-free interpolation of SSVI slices}.
\newblock {\em Decisions in Economics and Finance}, 42(2):665--677, 2019.

\bibitem{fukasawa2012normalizing}
M.~Fukasawa.
\newblock {The normalizing transformation of the implied volatility smile}.
\newblock {\em Mathematical Finance}, 22(4):753--762, 2012.

\bibitem{gatheral2004parsimonious}
Jim Gatheral.
\newblock {A parsimonious arbitrage-free implied volatility parameterization
  with application to the valuation of volatility derivatives}.
\newblock {\em Presentation at Global Derivatives \& Risk Management, Madrid},
  2004.

\bibitem{gatheral2014arbitrage}
Jim Gatheral and Antoine Jacquier.
\newblock {Arbitrage-free SVI volatility surfaces}.
\newblock {\em Quantitative Finance}, 14(1):59--71, 2014.

\bibitem{hendriks2019extended}
Sebas Hendriks and Claude Martini.
\newblock {The extended SSVI volatility surface}.
\newblock {\em Journal of Computational finance}, 22(5), 2019.

\bibitem{lee2004moment}
Roger~W Lee.
\newblock {The moment formula for implied volatility at extreme strikes}.
\newblock {\em Mathematical Finance: An International Journal of Mathematics,
  Statistics and Financial Economics}, 14(3):469--480, 2004.

\bibitem{martini2021explicit}
Claude Martini and Arianna Mingone.
\newblock {Explicit no arbitrage domain for sub-SVIs via reparametrization}.
\newblock {\em arXiv preprint arXiv:2106.02418}, 2021.

\bibitem{martini2021no}
Claude Martini and Arianna Mingone.
\newblock {No arbitrage SVI}.
\newblock {\em Available at SSRN 3594528, to appear in SIFIN 2022}, 2021.

\bibitem{tehranchi2020black}
Michael~R Tehranchi.
\newblock {A Black-Scholes inequality: applications and generalizations}.
\newblock {\em Finance and Stochastics}, 24(1):1--38, 2020.

\end{thebibliography}
\bibliographystyle{plain}

\end{document}